%% file: root.tex
\documentclass[letterpaper, 10 pt, conference]{ieeeconf}  

\IEEEoverridecommandlockouts                              

\usepackage{amsmath} 
\usepackage{amsthm}

\usepackage{amsfonts}
\usepackage{amssymb}  
\usepackage{mathtools} 
\usepackage{mathrsfs}
\usepackage{balance}
\DeclareMathAlphabet{\mathpzc}{OT1}{pzc}{m}{it}

\newcommand{\R}{\mathbb{R}}

\theoremstyle{definition}

\newtheorem{definition}{Definition}

\newtheorem{lemma}{Lemma}
\newtheorem{theorem}{Theorem}
\newtheorem{remark}{Remark}
\newtheorem{assumption}{Assumption}

\usepackage{cleveref}

\usepackage[ruled,linesnumbered]{algorithm2e}

\newcommand{\norm}[1]{\left\lVert #1\right\rVert}
\newcommand{\set}[1]{\left\{ #1\right\}}

\DeclareMathOperator*{\argmin}{arg\,min}

\newcommand{\doi}[1]{\href{http://dx.doi.org/#1}{\normalsize{\textsc{doi:}}~\nolinkurl{#1}}}
\newcommand{\arxiv}[1]{\href{http://arxiv.org/abs/#1}{\normalsize{\textsc{arxiv:}}~\nolinkurl{#1}}}
\usepackage{color}

\usepackage{comment}
\newcommand{\HRule}{\noindent\rule{\linewidth}{0.1mm}\newline}

\renewcommand{\epsilon}{\varepsilon}
\renewcommand{\phi}{\varphi}

\newcommand*\sq{\mathbin{\vcenter{\hbox{\rule{.4ex}{.4ex}}}}}

\title{\LARGE \bf
Gaussian Process-based Min-norm Stabilizing Controller for Control-Affine Systems with Uncertain Input Effects and Dynamics}
\author{Fernando Castañeda*, Jason J. Choi*, Bike Zhang, Claire J. Tomlin and Koushil Sreenath
\thanks{*Indicates equal contribution.}
\thanks{Fernando Castañeda, Jason J. Choi, Bike Zhang, Claire J. Tomlin and Koushil Sreenath are with the University of California, Berkeley, CA, 94720, USA, \tt\small \{fcastaneda, jason.choi, bikezhang, tomlin, koushils\}@berkeley.edu}
\thanks{This work was partially supported through National Science Foundation Grant CMMI-1931853 and DARPA Assured Autonomy program, grant FA8750-18-C-0101. The work of Fernando Casta\~neda received the support of a fellowship from Fundaci\'on Rafael del Pino, Spain. The work of Jason Choi received the support of a fellowship from Kwanjeong Educational Foundation, Korea.}
}

\begin{document}

\maketitle
\thispagestyle{empty}
\pagestyle{empty}

\begin{abstract}
This paper presents a method to design a min-norm Control Lyapunov Function (CLF)-based stabilizing controller for a control-affine system with uncertain dynamics using Gaussian Process (GP) regression. In order to estimate both state and input-dependent model uncertainty, we propose a novel compound kernel that captures the control-affine nature of the problem. Furthermore, by the use of GP Upper Confidence Bound analysis, we provide probabilistic bounds of the regression error, leading to the formulation of a CLF-based stability chance constraint which can be incorporated in a min-norm optimization problem. We show that this resulting optimization problem is convex, and we call it ``Gaussian Process-based Control Lyapunov Function Second-Order Cone Program'' (GP-CLF-SOCP). The data-collection process and the training of the GP regression model are carried out in an episodic learning fashion. We validate the proposed algorithm and controller in numerical simulations of an inverted pendulum and a kinematic bicycle model, resulting in stable trajectories which are very similar to the ones obtained if we actually knew the true plant dynamics.
\end{abstract}

\input{01introduction}
\input{02problemstatement}
\input{03gaussianprocesses}
\input{04compoundkernel}
\input{05optimization}
\input{06algorithm}

\input{07results}
\input{08conclusion}

\addtolength{\textheight}{-12cm}   




\section*{Acknowledgment}
We would like to express our gratitude to Alonso Marco for his insightful comments.

\balance
\bibliographystyle{IEEEtran}
\bibliography{reference.bib}{}
\end{document}

%% file: 01introduction.tex
\section{Introduction}
\label{sec:01introduction}

Model-based controllers have a problem inherent to their nature: model uncertainty. 
In this paper, we directly address this issue for the case of Lyapunov-based stabilizing controllers for nonlinear control-affine systems by using Gaussian Process (GP) regression to estimate the adverse effects of model uncertainty.

Control Lyapunov Functions (CLFs) \cite{artstein,sontag} have been widely used in recent years for nonlinear model-based stabilizing control of robotic systems \cite{galloway2015clfqp, nguyen2015optimal,reher2019inverseCLF}. Typically, the robot is stabilized by enforcing the CLF to decay to zero with a constraint in an optimization problem \cite{ames2013clfqp}.
However, CLF-based optimization methods heavily rely on the assumption that the model used for the controller design accurately represents the true plant's dynamics. If there is model-plant mismatch, convergence guarantees are often lost. Past research has directly addressed this issue by using both robust \cite{nguyen2015optimal} and adaptive \cite{nguyen2015l1adaptive} control theory.
More recently, various kinds of data-driven methods that use neural networks have been introduced \cite{taylor2019clflearning,choi2020reinforcement,westenbroek2020learning}.
Although these are demonstrated to be effective in practice, it is often difficult to verify the reliability of the neural network predictions.

For this paper, we are more interested in another class of data-driven approaches to tackle this problem, which use GP regression to allow for the analysis of the confidence of the prediction. The method of applying GPs to the CLF constraint was first introduced for closed-loop systems in \cite{berkenkamp2016lyapunov}. Then, similar approaches have also been proposed for the construction of stability and safety constraints to be incorporated in min-norm controllers \cite{umlauft2018clf, fan2019balsa, cheng2020safe, zheng2020learning}.

However, all of these papers make an important assumption that might restrict their applicability, which is that the considered model uncertainty is unaffected by the control input. In contrast, for many controlled systems, uncertain input effects\footnote{Uncertainty in the control vector field $g(x)$ in \eqref{eq:system}.} are prevalent, e.g., in a mechanical system, uncertainty in the inertia matrix directly induces uncertain input effects. In the work presented in \cite{khojasteh2020probabilistic}, a similar problem is addressed for the case of Control Barrier Function-based safety constraints \cite{amescbf} by the use of a Matrix-Variate GP regression. However, it does not provide a regression confidence analysis and results in an optimization problem that is not always convex. Finally, all the aforementioned GP-based approaches apply GP regression directly to the dynamics vector fields, which scale poorly with the system dimension.

In this paper, we develop solutions to overcome the presented limitations of the previous GP-based methods. First, we provide a formal way to deal with input-dependent model uncertainty of control-affine systems by proposing a specific GP kernel structure suitable for this problem. Since we apply GP regression to a scalar uncertainty term in the CLF constraint directly, compared to learning the uncertainty terms in the dynamics, we can reduce the computation of the regression significantly while still capturing many realistic forms of uncertainty. A similar kernel structure was used in \cite{umlauft2017feedback} to learn the uncertainty terms in the autonomous and control vector fields separately for a single-input system. Here, we generalize the kernel to an arbitrary input dimension and derive expressions for the posterior GP of a combined input-dependent uncertainty term whose mean and variance are linear and quadratic in the input, respectively.
By doing so, we can formulate a Second-Order Cone Program (SOCP) which incorporates a chance constraint that takes into account the confidence of the GP model and provides the exponential stabilizability of the system. We call it \emph{Gaussian Process-based Control Lyapunov Function Second-Order Cone Program} (GP-CLF-SOCP).
Formulation of the SOCP is crucial in that it can be solved quickly enough for real-time applications due to its convexity.
Finally, since the inference time of GP regression is directly determined by the size of the training data,
we maximize data efficiency by the use of an algorithm that iteratively collects data and improves the GP regression model in an episodic learning fashion.

The rest of the paper is organized as follows. In Section \ref{sec:02problemstatement}, we give a brief overview of CLF-based controllers and show the effects that model uncertainty has on them. In Section \ref{sec:03gp}, we explain the basic concepts of GP regression. In Section \ref{sec:04compoundkernel}, we present the compound kernel structure that allows us to regress the control-affine uncertainty. In Section \ref{sec:05optimization}, we present the SOCP formulation of our uncertainty-aware optimization problem. In Section \ref{sec:06algorithm}, we propose an efficient data-collection procedure. In Section \ref{sec:07results} we validate the proposed method for two different systems. Finally, in Section \ref{sec:08conclusion}, we provide concluding remarks.

%% file: 02problemstatement.tex
\section{Problem Statement}
\label{sec:02problemstatement}

Throughout this paper we will consider nonlinear control-affine systems of the form
\vspace{-3pt}
\begin{equation}\label{eq:system}
    \dot{x} = f(x) + g(x)u,
    \vspace{-3pt}
\end{equation}
where $x \in \R^n$ is the state of the system and $u \in \mathbb{R}^m$ is the control input. The vector fields $f \colon \R^n \to \R^n$ and $g \colon \R^n \to \R^{n \times m}$ are assumed to be locally Lipschitz continuous and $f(0)=0$.

The main objective of this paper is the construction of a locally stabilizing controller for such a system even when its dynamics are uncertain. A system is called \textit{stabilizable} when it is asymptotically controllable to the origin with a feedback control law $u \colon \R^n \to \R^m$ that is continuous except possibly at the origin.

\subsection{Control Lyapunov Functions}
\label{subsec:0201clf}

\begin{definition}\vspace{0.3cm}
\label{def:clf}
Let $V \colon \R^n \to \R_{+}$ be a positive definite, continuously differentiable and radially unbounded function. We say that $V$ is a \emph{Control Lyapunov Function} (CLF) for system \eqref{eq:system} if for each $x \in \mathbb{R}^n \setminus \set{0}$
\begin{equation}\label{eq:clf_def}
    \inf_{u \in \mathbb{R}^m} \underbrace{L_f V(x) + L_g V(x)u}_{= \dot{V}(x,u)} < 0,
\vspace{-3pt}
\end{equation}
where the functions $L_f V(x) \coloneqq \nabla V(x) \cdot f(x)$ and $L_g V(x) \coloneqq \nabla V(x) \cdot g(x)$ are known as Lie derivatives.
\end{definition}
If such a CLF exists, the system is known to be globally stabilizable \cite{artstein}. Then, it is desirable to find a locally Lipschitz continuous feedback control law $u \colon \R^n \to \R^m$ such that the condition $L_f V(x) + L_g V(x)u(x) < 0$ holds for any $x \in \R^n \setminus \set{0}$. A simple way of synthesizing such a control law is by enforcing \eqref{eq:clf_def} as a constraint in a min-norm optimization problem. 
If $u$ is unconstrained, this min-norm stabilizing controller can be expressed in closed-form \cite{sontag}.

However, many real-world systems require the addition of input constraints due to actuator limitations, i.e., $u \in U \subset \R^m$, in which case condition \eqref{eq:clf_def} becomes $\inf_{u \in U} L_f V(x) + L_g V(x)u < 0$,
which might not be satisfied at every $x \in \mathbb{R}^n \setminus \set{0}$ even if $V$ is a valid CLF for system \eqref{eq:system}. This fact motivates the following lemma.

\begin{lemma}\vspace{0.3cm}
\label{th:clf}
Let $V \colon \R^n \to \R_{+}$ be a CLF for system \eqref{eq:system} and let $U \subset \R^m$ be the compact set of admissible control inputs. For each $c \in \R_{+}$ let $\Omega_c$ be the sublevel set of $V$ such that $\Omega_c \coloneqq \set{x \in \R^n \colon V(x) \leq c}$. If there exists a $c_i > 0$ such that
\vspace{-5pt}
\begin{equation}
\label{eq:clf_local_def}
    \inf_{u \in U} L_f V(x) + L_g V(x)u < 0
    \vspace{-2pt}
\end{equation}
is satisfied $\forall x \in \Omega_{c_i} \setminus \set{0}$, then the system is locally stabilizable and $\Omega_{c_i} $ is a control invariant subset of the Region of Attraction (RoA) of the origin.
\end{lemma}
\begin{proof}
See \cite[Proposition~2.2]{Lin95control-lyapunovuniversal} for the proof of stabilizability. The control invariance proof is straightforward since for any $x$ at the boundary of $\Omega_{c} $ we can always find a $u\in U$ such that $\dot{V}(x,u)<0$ from condition \eqref{eq:clf_local_def}.
\end{proof}

We can now take $c_{max}$ as the maximum value of $c_i \in \R_{+}$ such that \eqref{eq:clf_local_def} holds for any $x \in \Omega_{c_i}  \setminus \set{0}$. Then, $\Omega_{c_{max}} $ is the largest sublevel set of $V$ contained in the RoA.

We can also consider a stronger notion of stabilizability by imposing exponential convergence to the origin. 
It is well-known that if there exists a compact subset $D \subseteq \Omega_{c_{max}} $ such that $\forall x \in D$ the following holds for some constant $\lambda>0$,
\begin{equation}
\label{eq:expclf_local_def}
    \inf_{u \in U} L_f V(x) + L_g V(x)u + \lambda V(x) \leq 0,
\end{equation}
 then the state of the system can be driven exponentially fast to the origin from any initial state $x_0 \in \Omega_{c_{exp}} \subseteq D$ \cite{ames2014rapidly}. If such $c_{exp}>0$ exists, we will say that $V$ is a \textit{locally exponentially stabilizing} CLF. The condition
\eqref{eq:expclf_local_def} can be incorporated as a constraint into a min-norm optimization problem:
\HRule
\noindent \textbf{CLF-QP}:
\begin{subequations}
\label{eq:clf-qp-all}
\begin{align}
u^{*}(x) & = & & \underset{u\in U}{\argmin} \quad u^T u \label{eq:clf-qp}\\
\text{s.t.} & \; & & L_f V(x) + L_g V(x)u + \lambda V(x) \leq 0. \label{eq:eclf}
\end{align}
\end{subequations}
\HRule
In this paper, we assume that the input constraints are linear, which makes problem \eqref{eq:clf-qp-all} a quadratic program (QP). We will refer to constraint \eqref{eq:eclf} as the \textit{exponential CLF constraint}. This optimization problem defines a feedback control law $u^*\!\colon\!\R^n\!\to\!\R^m$ selecting the min-norm input such that the system state converges to the origin exponentially quickly. 
Note that, in practice, constraint \eqref{eq:eclf} is typically relaxed by adding a slack variable in order to guarantee the feasibility of the problem if condition \eqref{eq:expclf_local_def} is not locally satisfied \cite{galloway2015clfqp}.

\subsection{Effects of Model Uncertainty on CLF-based Controllers}
\label{subsec:02adverse-effects}

The main problem concerned in this paper is how to reformulate the min-norm stabilizing controller defined in \eqref{eq:clf-qp-all} in the presence of model uncertainty. 

First, we provide some necessary settings and assumptions for our problem formulation.
Let's assume that we have a \textit{nominal model}
\begin{equation}
    \label{eq:model}
    \dot{x} = \tilde{f}(x) + \tilde{g}(x)u,
\end{equation}
where $\tilde{f}: \R^n \rightarrow \R^n$, $\tilde{g}: \R^n \rightarrow \R^{n\times m}$ are Lipschitz continuous vector fields and $\tilde{f}(0)=0$. We assume that we have a locally exponentially stabilizing CLF $V$ for the nominal model, and that the plant is also locally exponentially stabilizable with the same $V$. Note that the region of exponential stabilizablity around the origin can be sufficiently small. Also, the assumption can be relaxed to asymptotic stabilizability if the user is concerned with enforcing condition \eqref{eq:clf_local_def} instead of \eqref{eq:expclf_local_def}. In general, $\tilde{f}$ and $\tilde{g}$ would be different from the true plant vector fields $f$ and $g$ because the nominal model is imperfect. The assumption implies, however, that they share some similarity through the stabilizing property of the same function $V$. Finally, we also assume that we have access to measurements of state and control input at every sampling time $\Delta t$.

Our main objective is to construct the exponential CLF constraint \eqref{eq:eclf} for the true plant when we only know the model dynamics $\tilde{f}$ and $\tilde{g}$. Since $\dot{V}(x,u)=L_f V(x) + L_g V(x)u$ depends on the dynamics of the plant, the estimate based on the nominal model $\tilde{\dot{V}}(x, u)=L_{\tilde{f}}V(x) + L_{\tilde{g}}V(x) u$, will differ from its true value. We define $\Delta: \R^n \times \R^m \rightarrow \R$ as the difference between these:
\vspace{-3pt}
\begin{equation}
    \label{eq:mismatchdef}
    \Delta(x, u) := \dot{V}(x, u) - \tilde{\dot{V}}(x, u).
    \vspace{-2pt}
\end{equation}

\noindent Then, the exponential CLF constraint \eqref{eq:eclf} becomes
\vspace{-2pt}
\begin{equation}
    \label{eq:clf-constraint-uncertainty}
    L_{\tilde{f}}V(x) + L_{\tilde{g}}V(x) u + \Delta(x, u) + \lambda V(x)\le 0.
\vspace{-2pt}
\end{equation}
Therefore, verifying the exponential CLF constraint for the true plant amounts to a problem of learning the mismatch term $\Delta(x, u)$ correctly and then enforcing \eqref{eq:clf-constraint-uncertainty}. 
We can learn this function from the past data by formulating a supervised learning problem. Specifically, we will use GP regression, a method that will be introduced in the next section.

\begin{remark}
In \eqref{eq:mismatchdef}, if we express $\dot{V}$ and $\tilde{\dot{V}}$ with their respective Lie derivatives, we get
\vspace{-3pt}
\begin{equation}
\label{eq:mismatchaffine}
    \Delta(x, u) = \underbrace{(L_f V(x) -L_{\tilde{f}}V(x))}_{=:\Delta_1(x)} + \underbrace{(L_g V(x) - L_{\tilde{g}}V(x) )}_{=:\Delta_2(x)}u.
    \vspace{-3pt}
\end{equation}
Note that we do not have access to $\Delta_1(x)$ and $\Delta_2(x)$ in this equation since we are unaware of $f$ and $g$. It is tempting to learn each of these terms separately with supervised learning. However, we can only measure $\Delta(x, u)$, which makes this approach intractable. Nevertheless, we can exploit the fact that the mismatch term $\Delta(x, u)$ is control-affine.
\end{remark}

%% file: 03gaussianprocesses.tex
\section{Gaussian Process Regression}

\label{sec:03gp}
A Gaussian Process is a random process such that any finite selection of samples $\{h(x_k)\}_{k=1}^{n}$ has a joint Gaussian distribution. A GP is fully determined by its mean function $q: \mathcal{X}\rightarrow\R$ and covariance function $k: \mathcal{X}\times\mathcal{X}\rightarrow\R$, i.e.,
\begin{equation}
    h(x) \sim \mathcal{GP}(q(x), k(x,x')),
\end{equation}
where $\mathcal{X}$ is the input domain, a connected subset of $\R^n$, $h(x)$ is the output function sampled from the GP, and $x, x'$ denote input variables in $\mathcal{X}$. Any positive definite kernel function\footnote{ $k$ is a positive definite kernel if its associated kernel matrix 
$K(x_1, x_2)$, whose ($i^{th}$, $j^{th}$) element is defined as $k(x_i, x_j)$, is positive semi-definite for any distinct points $x_1, x_2 \in \mathcal{X}$.
} can be a valid covariance function \cite{wendland2004scattered}.
Such a kernel $k(x, x')$ can be used to generate a set of functions that satisfy a specific property, namely a ``reproducing" property. An inner product between such a function $h$ and the kernel $k(\cdot, x)$ should reproduce $h$, i.e., $\langle h(\cdot), k(\cdot, x) \rangle = h(x), ~\forall x\in \mathcal{X}$. Such a set is called a Reproducing Kernel Hilbert Space (RKHS, \cite{wendland2004scattered}), a specific class of function space, and is denoted as $\mathcal{H}_k(\mathcal{X})$. The RKHS norm $\norm{h}_{k}:=\sqrt{\langle h, h \rangle}$, which will be used in Lemma \ref{lemma:UCB}, is a measure of the smoothness of $h$ with respect to the kernel function\footnote{$\norm{h(x)-h(x')}_{2}\le\norm{h}_{k}\norm{k(x,\cdot)-k(x',\cdot)}_{k} \; \forall x, x' \in \mathcal{X}$}. Note that an appropriate inner product in the above expressions would be determined by the specific choice of the associated reproducing kernel $k$.

GPs encode prior distributions over functions and given new query points, a posterior distribution can be derived from the joint Gaussian distribution between the prior data and the query points. This gives rise to their typical application in the machine learning literature: GP regression. For the remainder of the paper, we use $q(x)\!\equiv\!0$ as the mean function of the prior GP.
Given $N$ input-output data pairs $\{(x_j,z_j)\}_{j=1}^{N}$, the regressor is provided by the posterior GP distribution conditioned on the data. Here, the output data is assumed to be measurements of $h(x_j)$, where $h$ is the target function for regression, with additive white noise, i.e., $z_j = h(x_j)+\epsilon_j$, where $\epsilon_j \sim \mathcal{N}(0, \sigma_{n}^2)$. Then, the mean and the variance of the posterior $h(x_{*})$ at a query point $x_*$, are given as 
\begin{equation}
\label{eq:gpposteriormu}
        \mu_{*} = \mathbf{z}^T (K + \sigma_n^2 I )^{-1} K_{*}^{T},
\end{equation}
\begin{equation}
\label{eq:gpposteriorsigma}
    \sigma_{*}^{2} = k\left(x_{*}, x_{*}\right)-K_{*}  (K + \sigma_n^2 I )^{-1} K_{*}^{T},
\end{equation}
which are derived from the distribution of $h(x_{*})$ conditioned on $\{z_j\}_{j=1}^{N}$ \cite{williams2006gaussian}.
$K\in\R^{N\times N}$ is the Gram matrix whose ($i^{th}$, $j^{th}$) element is defined as $k(x_i, x_j)$ for $i, j=1,\cdots,N$,
and $K_{*}=[k(x_*, x_1),\ \cdots \ ,k(x_*, x_N)]\in\R^{1\times N}$. $\mathbf{z}\in \R^N$ is the vector containing the output measurements $z_j$. 
Note that there exist various choices of kernel functions and many of them depend on some hyperparameters which determine the kernel's characteristics. Depending on the choice of kernel and hyperparameters, the result of the regression varies, and the problem of choosing the best kernel and its hyperparameters is known as the ``training" process of the GP regression \cite{williams2006gaussian}. In this work we use marginal likelihood maximization, which is one of the most common training methods.

After training, one would like to study how close the GP model approximates the target function. In order to do this, we use the Upper Confidence Bound (UCB) analysis \cite{gpucb}, specifically, the following lemma.

\begin{lemma}\cite[Thm.~6]{gpucb} \label{lemma:UCB} Assume that the noise sequence $\{\epsilon_j\}_{j=1}^{\infty}$ is zero-mean and uniformly bounded by $\sigma_{n}$.
Let the target function $h: \mathcal{X} \rightarrow \mathbb{R}$ be a member of $\mathcal{H}_{k}(\mathcal{X})$ associated with a bounded kernel $k$, with its RKHS norm bounded by $B$. Then, with probability of at least $1 - \delta$, the following holds for all $x \in \mathcal{X}$ and $N \geq 1$:

\vspace{-10pt}
\small{
\begin{equation*}
    | \mu_{*} - h(x_{*}) | \leq \left(2B^2 + 300 \gamma_{N+1} \ln^3((N+1)/\delta)\right)^{0.5} \sigma_{*},
\end{equation*}} 

\vspace{-10pt}
\normalsize
\noindent where $\gamma_{N+1}$ is the maximum information gain after getting $N+1$ data points, and $\mu_{*}$, $\sigma^2_{*}$ are the mean and variance of the posterior GP given by \eqref{eq:gpposteriormu} and \eqref{eq:gpposteriorsigma}.

\end{lemma}

\begin{proof}
See \cite[Thm.~6]{gpucb}.
\end{proof}

In this lemma, the assumption about the boundedness of $\norm{h}_{k}$ implicitly requires a ``low complexity" of the target function \cite{chowdhury2017newgpucb}. $B$ is usually unknown a priori, but a trial-and-error approach to find its value suffices in practice \cite{gpucb}. $\gamma_{N+1}$ quantifies the reduction of uncertainty about $h$ in terms of entropy. It has a sublinear dependency on $N$ for many commonly used kernels and it can be efficiently approximated up to a constant \cite{gpucb}. 

%% file: 04compoundkernel.tex
\section{GP Regression for Affine Target Functions}
\label{sec:04compoundkernel}

In this section, we use GP regression to learn the mismatch term $\Delta(x, u)$ \eqref{eq:mismatchdef} from data.
From \eqref{eq:mismatchaffine}, we know that $\Delta(x, u)$ is affine in $u$. If we use an arbitrary kernel, we cannot exploit this information in the GP regression.
Therefore, our first objective is to construct an appropriate kernel that captures the control-affine structure of $\Delta$ in the regression. In order to do this, we introduce the general formulation of this problem in this section. Consider $p$ functions, $h_i:\mathcal{X} \rightarrow \R$ for $i = 1,\cdots,p$, and define
\vspace{-3pt}
\begin{equation}
    h_{c}(x, y) \coloneqq [h_1(x) \; h_2(x) \; \cdots \; h_p(x)]\cdot y,
\label{eq:h_c}
\vspace{-2pt}
\end{equation}
where $y\in\mathcal{Y}\subset\R^p$. Our objective is to estimate the function $h_{c}:\mathcal{X} \times \mathcal{Y} \rightarrow \R$ which is affine in $y$ by using GP regression, given its measurements $z_j = h_{c}(x_j, y_j) + \epsilon_j$ for $j = 1, \cdots, N$.

The underlying structure of $h_c (x, y)$ tells us that it contains information about $p$ random functions $\{h_i (x)\}_{i=1}^p$ condensed to a single scalar value by a dot product with $y$. Therefore, it is natural to consider $p$ underlying kernels and their composition. For $i=1, \cdots, p$, consider covariance functions $k_{i}:\mathcal{X} \times \mathcal{X} \rightarrow \R$.
\begin{definition} \label{def:adpkernel}
\emph{Affine Dot Product Compound Kernel: }
Define $k_{c}$ given by
\small
\begin{equation}
    k_{c}\left(\left[\begin{array}{c} x \\ y \\\end{array}\right], \left[\begin{array}{c} x' \\ y' \\\end{array}\right]\right) := y^T Diag([k_1(x, x'), \cdots, k_p(x, x')]) y',
\label{eq:adpkernel}
\end{equation}
\normalsize

\noindent as the \emph{Affine Dot Product} (ADP) compound kernel of $p$ individual kernels $k_1(x, x'),\ \cdots \ ,$ $k_p(x, x')$.
\end{definition}

Note that for a fixed $(x, x')$, the ADP compound kernel resembles the well-known dot product kernel, defined as $k(y, y') = y^T y'$ \cite{williams2006gaussian}.

\begin{lemma}
\label{lemma:adpkernelcharacter}
If $k_1, \cdots , k_p$ are positive definite kernels, the ADP compound kernel $k_c$ is also positive definite. Furthermore, if $k_1, \cdots , k_p$ are bounded kernels, $k_c$ is also bounded.
\end{lemma}
\begin{proof}
Consider the Gram matrix of $k_c$, $K_c\in \R^{N\times N}$ for $\{(x_j, y_j)\}_{j=1}^{N}$. Let $K_i$ be the Gram matrix of $k_i$ for $\{x_j\}_{j=1}^{N}$. Define $Y := [y_1\; y_2\; \cdots\; y_N] \in \R^{p\times N}$, and let $\mathbf{y_{i}}^T$ be the $i$-th row of $Y$. Then,
\begin{equation*}
    K_c = \sum_{i=1}^{p}{\left(\mathbf{y_{i}}\;\mathbf{y_{i}}^T\right) \circ K_i },
\end{equation*}
where $\circ$ indicates the Hadamard product \cite{horn1990hadamard}. By the Schur Product Theorem \cite{horn1990hadamard}, if the $k_i$ are positive definite kernels, then since each $\mathbf{y_{i}}\mathbf{y_{i}}^T$ and $K_i$ are positive semidefinite, $K_c$ is a positive semidefinite matrix. Therefore, $k_c$ is a positive definite kernel by definition. Also, if the $k_i$ are bounded kernels, each $K_i$ is bounded so $K_c$ is also bounded. Therefore, $k_c$ is a bounded kernel.
\end{proof}

By Lemma \ref{lemma:adpkernelcharacter}, since $k_c$ is positive definite, it is a valid covariance function. Consider a set of functions $\mathcal{H}_{k_{c}}(\mathcal{X}\times \mathcal{Y}):=\{h_c:\mathcal{X}\!\times\!\mathcal{Y}\!\rightarrow\!\R \;|\; \exists h_i\!\in\!\mathcal{H}_{k_{i}} \;\text{for}\; i=1,\cdots,p,\;$s.t.$\;h_{c}(x, y) =  [h_1(x), \cdots, h_p(x)] \cdot y \}$ where each $\mathcal{H}_{k_{i}}$ is the RKHS whose reproducing kernel is $k_i$. Then, the following holds: 
\begin{theorem}
\label{th:RKHSforADPkernel}
$\mathcal{H}_{k_{c}}(\mathcal{X}\times \mathcal{Y})$ is an RKHS whose reproducing kernel is $k_c$ in Definition \ref{def:adpkernel}.
\begin{proof}
Define the inner product of $\mathcal{H}_{k_{c}}$ to be
\vspace{-5pt}
\begin{equation*}
    \langle h_c, h_{c}'\rangle_{c}:=\sum_{i=1}^p {\langle h_{i}, h_{i}'\rangle_{i}},
\end{equation*}
for $\forall h_c,h_{c}' \in \mathcal{H}_{k_{c}}$ where $\{h_i\}_{i=1}^p$ and $\{h_{i}'\}_{i=1}^p$ are sets of functions whose $i$-th elements are from $\mathcal{H}_{k_{i}}$ that satisfy $h_{c}(x, y) =  [h_1(x), \cdots, h_{p}(x)] \cdot y$ and $h_{c}'(x, y) =  [h_{1}'(x), \cdots, h_{p}'(x)] \cdot y$, respectively. Such sets of functions should exist by definition of $\mathcal{H}_{k_{c}}$. $\langle h_{i}, h_{i}'\rangle_{i}$ is the inner product of $\mathcal{H}_{k_{i}}$. It is trivial that this definition satisfies the axioms of the inner product. Then,

\vspace{-10pt}
\small
\begin{align*}
    \left\langle h_{c}(\cdot, \sq), k_{c}\left(\begin{bmatrix} \cdot \\ \sq \end{bmatrix}, \left[\begin{array}{c} x \\ y \\\end{array}\right]\right) \right\rangle_{c} 
    = & \sum_{i=1}^p y_{i} \langle h_{i}(\cdot), k_{i}(\cdot, x)\rangle_{i} \\
    = & \sum_{i=1}^p y_{i} h_{i}(x) = h_{c}(x, y).
\end{align*}
\normalsize

The first equality holds because of Definition \ref{def:adpkernel} and the definition of $\langle \cdot, \cdot \rangle_{c}$. The second equality holds because of the reproducing property of each $k_{i}(\cdot, \cdot)$.
\end{proof}
\end{theorem}

Theorem \ref{th:RKHSforADPkernel} allows us to apply the UCB analysis from Section \ref{sec:03gp} to $h_{c}(x, y)$ with some additional conditions which will be specified in Section \ref{sec:05optimization}. Regression for $h_{c}(x, y)$ in Lemma \ref{lemma:UCB} (i.e. $\mu_{*},\sigma_{*}$) now can be treated in the same way as any other kind of general GP regression, but with a specific choice of covariance function given by \eqref{eq:adpkernel}.

One caveat of this regression is that depending on the distribution of the inputs $y_j$ in the data, this problem can be underdetermined. For instance, when every $y_j$ is a constant vector, there are infinitely many choices of valid $h_i(x)$ that give the same estimation error. Nevertheless, under our GP regression structure, this evidence of underdetermination is implicitly captured by larger values of the variance of the posterior. In practice, it is preferable to avoid such underdetermination since we want to reduce the uncertainty of the GP posterior. Therefore, we need to carefully collect the training data to make sure we capture rich enough information about the target function. In Section \ref{sec:06algorithm} we propose a method for this purpose. In the system identification literature, this is related to the property of persistency of excitation \cite{verhaegen2007filtering}.

Finally, the main benefit of exploiting the affine structure in the kernel is revealed in the expressions for the posterior distribution's mean and variance. This is the main difference in how we use the ADP kernel compared to \cite{umlauft2017feedback}, where a similar kernel is proposed for a special case $p=2$. Let $X\in\R^{n\times N},\;Y\in \R^{p \times N}$ be matrices whose column vectors are the inputs $x_j$ and $y_j$ of the collected data, respectively, and let $\mathbf{z}\in \R^N$ be the vector containing the output measurements $z_j$.
Then, plugging them and the ADP compound kernel into \eqref{eq:gpposteriormu} and \eqref{eq:gpposteriorsigma} gives the following expressions for the mean and variance of the posterior at a query point $(x_{*}, y_{*})$:
\vspace{-3pt}
\begin{equation}
\label{eq:mu_adp}
    \mu_{*} = \underbrace{\mathbf{z}^T (K_c + \sigma_n^2 I )^{-1} K_{*Y}^{T}}_{=:b_{*}^T} y_{*},
\end{equation}
\begin{equation}
\label{eq:sigma_adp}
\small{
    \sigma_{*}^{2} = y_{*}^{T}\!\underbrace{\left(Diag \text{\footnotesize $\left(\begin{bmatrix}k_1(x_{*}, x_{*}) \\ \vdots \\ k_p(x_{*}, x_{*}) \end{bmatrix}\right)$} -K_{*Y}  (K_c + \sigma_n^2 I )^{-1} K_{*Y}^{T} \right)}_{=:C_{*}}\!y_{*}.
}
\end{equation}
Here, $K_c\in\R^{N\times N}$ is the Gram matrix of $k_c$ for the training data inputs ($X,Y$), and $K_{*Y}\in\R^{p\times N}$ is given by
\begin{equation*}
\small{
    K_{*Y} = \begin{bmatrix} K_{1*} \\ K_{2*} \\ \vdots \\K_{p*}
    \end{bmatrix} \circ Y,\; K_{i*}=[k_i(x_{*}, x_1),\ \cdots \ , k_i(x_{*}, x_N)].
}
\end{equation*}

Readers can observe that \eqref{eq:mu_adp} and \eqref{eq:sigma_adp} are affine and quadratic in $y_{*}$, respectively. These structures are critical when formulating the uncertainty-aware CLF chance constraint as a second-order cone constraint in the next section.

%% file: 05optimization.tex
\section{Uncertainty-Aware Min-norm Stabilizing Controller}
\label{sec:05optimization}
\subsection{Probabilistic Bounds on the CLF Derivative}
We have already presented all the necessary tools to verify the probabilistic bounds on the mismatch term $\Delta(x,u)$ in \eqref{eq:mismatchaffine}. 
Indeed, learning $\Delta$ corresponds to the GP regression problem defined by \eqref{eq:mu_adp}, \eqref{eq:sigma_adp}, in which the target function $h_c$ is $\Delta$, $x$ is the state, $y=[1 , u^T ]^T $, $p=m+1$, $h_1$ is $\Delta_1$, and $h_{i+1}$ is $\Delta_2$'s $i$-th element for $i=1,\cdots,m$.

\begin{assumption}
\label{assum:RKHSnorm} Consider bounded reproducing kernels $k_i$ for $i=1,\cdots,m\!+\!1$. We assume that $\Delta_1$ is a member of $\mathcal{H}_{k_{1}}$ and each $i$-th element of $\Delta_2$ is a member of $\mathcal{H}_{k_{i+1}}$ for $i=1,\cdots,m$, respectively. We assume that their RKHS norms are bounded.

\end{assumption}
\begin{lemma}
\label{lemma:bounded-kernels} Under Assumption \ref{assum:RKHSnorm} and with a compact set of admissible control inputs $U$, $\Delta$ is a member of $\mathcal{H}_{k_{c}}$, the RKHS created by the ADP compound kernel of $k_i$ for $i=1,\cdots,m\!+\!1$. Moreover, its RKHS norm is bounded, namely $\norm{\Delta}_{k_c} \le B$. 
\end{lemma}
\begin{proof}
The proof follows from Thm. \ref{th:RKHSforADPkernel} and the definition of the inner product for $\mathcal{H}_{k_{c}}$ in the proof of Thm. \ref{th:RKHSforADPkernel}.
\end{proof}
\begin{assumption}
\label{assum:epsilon} We have access to measurements $z_i = \dot{V}(x_i,u_i) - (L_{\tilde{f}}V(x) + L_{\tilde{g}}V(x) u_i) + \epsilon_i$, and the noise term $\epsilon_i$ is zero-mean and uniformly bounded by $\sigma_{n}$.
\end{assumption}

With Assumptions \ref{assum:RKHSnorm}, \ref{assum:epsilon} and Lemma \ref{lemma:bounded-kernels}, we can now apply Lemma \ref{lemma:UCB} to our regression problem.

\begin{theorem}
\label{theorem:DeltaUCB}
Let Assumptions \ref{assum:RKHSnorm} and \ref{assum:epsilon} hold. Let $\beta \coloneqq \left(2B^2 + 300 \gamma_{N+1} \ln^3((N+1)/\delta)\right)^{0.5}$, with $N$ the number of data points, and $\gamma_{N+1}$ as defined in Lemma \ref{lemma:UCB}. Let $\mu_{*}$ and $\sigma^2_{*}$ be the mean and variance of the posterior for $\Delta$ using the ADP compound kernel, at a query point $(x_{*},u_{*})$ as obtained from \eqref{eq:mu_adp} and \eqref{eq:sigma_adp}. Then, with a probability of at least $1-\delta$ the following holds:
\begin{equation}
\label{eq:sigmaUCB}
    | \mu_{*} - \Delta(x_{*},u_{*}) | \leq \beta \sigma_{*}.
\end{equation}
\end{theorem}
\begin{proof}
Proof follows from Lemmas \ref{lemma:UCB} and \ref{lemma:bounded-kernels}.
\end{proof}
\vspace{-3pt}
The error in the estimation of the mismatch term $\Delta$ is now bounded for some confidence level.
From \eqref{eq:sigmaUCB} we can easily derive the bounds on the true derivative of the CLF for a probability of at least $1-\delta$:

\vspace{-5pt}
\small
\begin{equation}
\label{eq:V_UCB}
    \tilde{\dot{V}}(x_*,u_*) + \mu_{*} - \beta \sigma_{*} \leq \dot{V}(x_*,u_*) \leq \tilde{\dot{V}}(x_*,u_*) + \mu_{*} + \beta \sigma_{*}.
\end{equation}
\normalsize

\subsection{GP-Based CLF Second-Order Cone Program}

Taking the upper bound of \eqref{eq:V_UCB}, we can enforce the exponential CLF constraint of \eqref{eq:eclf} with a probability of at least $1-\delta$, and incorporate the resulting chance constraint into a min-norm optimization problem that defines a feedback control law 
$u^* \colon \R^n \to \R^m$ pointwise:

\small

\HRule
\noindent \textbf{GP-CLF-SOCP}:
\begin{align}
u^{*}(x) & = & & \underset{u \in U,\ d \in \R}{\argmin} \quad u^T u + p~d^2 \label{eq:gp-clf-socp} \\
\text{s.t.} & \; & & \tilde{\dot{V}}(x,u) + \mu_{*}(x,u) + \beta \sigma_{*}(x,u) + \lambda V(x) \leq d \nonumber.
\end{align}
\HRule

\normalsize

\vspace{-3pt}
\noindent With a slight abuse of notation, $\mu_{*}(x,u)$ and $\sigma_{*}(x,u)$ are the mean and standard deviation of the GP posterior at $(x,u)$, obtained from \eqref{eq:mu_adp} and \eqref{eq:sigma_adp}.

\begin{remark}
The stability constraint is relaxed in order to guarantee the feasibility of the problem. If the initial state $x_0$ is outside the CLF maximum sublevel set for exponential stability $\Omega_{c_{exp}}$, we cannot guarantee exponential convergence and neither can the controller which uses the true plant dynamics. However, even for this case, we still do guarantee that the approximation error of the CLF derivative is bounded as given by \eqref{eq:V_UCB} with probability $1-\delta$. 
\end{remark}

Note that this optimization problem does not require knowledge about the true plant dynamics. The fact that $\mu_{*}$ and $\sigma^2_{*}$ are affine and quadratic in $u$, respectively, is crucial for the following main result of the paper:

\begin{theorem}
\label{theorem:SOCP}
Using the proposed ADP compound kernel from Definition \ref{def:adpkernel}, the uncertainty-aware optimization problem \eqref{eq:gp-clf-socp} is convex, meaning that its global minimum can be reliably recovered. Specifically, it is a Second-Order Cone Program (SOCP).
\end{theorem}

\begin{proof}
Let's first transform the quadratic objective function into a second-order cone constraint and a linear objective.
Let the objective function be $J(u,d) \coloneqq u^T u + p~d^2$. Note that by taking $\phi = [u^T,d]^T$ we can express the objective as $J(\phi) = \phi^T Q \phi$. Next, by setting $z \coloneqq L \phi$, where $L$ is the matrix square-root of $Q$, we can rewrite $J(z) = \norm{z}_2^2$. Note that minimizing $J$ gives the same result as minimizing $J'(z) \coloneqq \norm{z}_2$. Now we can move the objective function $J'$ into a second-order cone constraint by setting $\norm{z}_2 \leq t$ and minimizing the new linear objective function $J''(t) \coloneqq t$.

The next step is to prove that the CLF chance constraint is a second-order cone constraint. Note that $\tilde{\dot{V}}(x, u)=L_{\tilde{f}}V(x) + L_{\tilde{g}}V(x) u$ and $\mu_{*}(x,u) = b_{*}^T [1, u^T]^T$ are both control-affine. Note that $\sigma_{*}(x,u) = \sqrt{[1, u^T]C_{*}(x)[1, u^T]^T}$ can be rewritten as $\sigma_{*}(x,u) = \norm{M(x) u + n(x)}_2$, although we omit the expressions of $M$ and $n$ for conciseness. Therefore, the CLF chance constraint is a second-order cone constraint, and the resulting optimization problem is an SOCP with two second-order cone constraints corresponding to the original objective function and the CLF chance constraint. SOCPs are inherently convex.
\end{proof}

%% file: 06algorithm.tex
\section{Data Collection}
\label{sec:06algorithm}

In this section, we introduce an algorithm that efficiently collects measurements of $\Delta$ for the GP regression.
This data should contain rich enough information about $\Delta$, especially about its dependency on $u$ as discussed in Section \ref{sec:04compoundkernel}, and since our goal is to obtain a locally stabilizing controller, it is preferable to exclude the data from outside the RoA for efficiency. To this end, we propose an algorithm that iteratively collects new data and trains a new GP model in an episodic learning fashion. The algorithm uses the level sets of the CLF as ``guides" for expanding the training region by exploiting Lemma \ref{th:clf}. In addition, we use the idea of greedy search in the Bayesian Optimization literature \cite{gpucb} to actively explore the most uncertain area of the training region. Our algorithm is based on the active learning algorithm of \cite{berkenkamp2017saferl}, although while \cite{berkenkamp2017saferl} focuses on guaranteeing safety online, our objective is to maximize the efficiency of the offline data collection.

\subsection{Discrete-Time Measurements}
First consider how to obtain inputs $(x_j, u_j)$ and labels ($z_j$) ---measurements of $\Delta(x_j, u_j)$--- of the training data. Let $x(t)$ and $u(t)$ be the state and control input measurements at time $t$ and $x(t+\small{\Delta}t)$ be the state measurement at the next timestep. We can use these values to create input-label pairs with $\mathcal{O}(\small{\Delta}t^2)$ approximation error:

\vspace{-5pt}
\small
\begin{align*}
    & x_j = \frac{x(t+\small{\Delta}t) + x(t)}{2}, \quad u_j=u(t), \\ 
    & z_j = \frac{V(x(t+\small{\Delta}t))-V(x(t))}{\small{\Delta}t} - \tilde{\dot{V}}(x_j, u_j).
\end{align*}

\normalsize

\noindent Note that $u_j$ is the control input during the interval $[t, t+\small{\Delta}t)$, and $z_j$ is the difference between the value of $\dot{V}(x_j, u_j)$ obtained from numerical differentiation and the nominal model-based estimate.

\begin{figure*}
\begin{center}
\includegraphics[width=\textwidth]{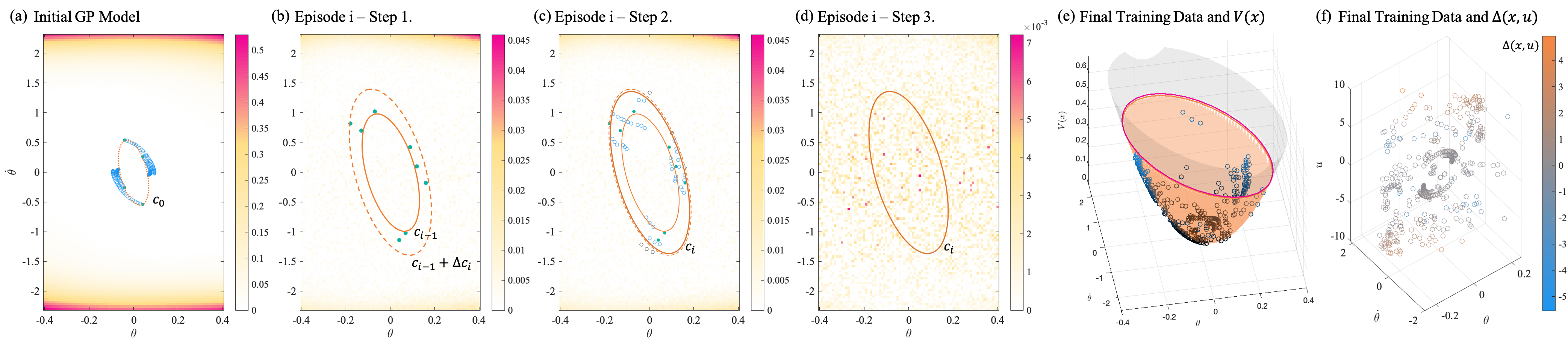}
\end{center}
\vspace{-1em}
\caption{\textbf{(a--d): Visualization of the episodic learning data collection algorithm running on the inverted pendulum example:} Color map represents the maximum variance of the posterior GP, $\max_{u\in U}\Delta_* (x, u)$. Orange curves: level curves of the CLF. Green points: initial states for the rollouts, Blue points: trajectory points added to the training data. Grey points: trajectory points excluded from the training data since they are outside $\Omega_{c_{i}}$. \textbf{(a) Initial GP Model: }Trajectories sampled from the initial level set $\Omega_{c_0}$ by running the CLF-QP are collected to create an initial GP model. \textbf{(b) Episode $i$-Step 1: }$N_e$ initial states and initial control inputs in $(\Omega_{ (c_{i-1} + \Delta c_i)} \setminus \Omega_{c_{i-1}} ) \times \mathcal{U}$ are determined where $\sigma_*$ are maximal. \textbf{(c) Episode $i$-Step 2: }Simulations are run from such initial points and the resulting trajectories are saved. At the same time, $c_i$ is determined by evaluating \eqref{eq:learnedLyapCert} for the sampled trajectories. \textbf{(d) Episode $i$-Step 3: }Finally, the $i$-th GP model is updated. Note the reduction in the variance. (Total episodes = 7, $i=3$ for (b), (c), (d).)
\textbf{(e, f): Distribution of the final training data }plotted in the $x$--$V(x)$ space (blue points) and plotted in $x$--$u$ space, respectively. (e) Level curve in color magenta is the $\Omega_{c_{max}}$ (maximum level set contained in the RoA) for the true plant. The value of CLF is plotted in grey and the orange region is the region verified as RoA through the data collection algorithm. (f) The color indicates the value of $z_i$, the measurement of $\Delta(x_{i}, u_{i})$. The number of data points is 425.}\label{fig:algorithm}
\vspace{-1em}
\end{figure*}

\subsection{Estimation of the Region of Attraction}

Next, we introduce a new certificate with the learned uncertainty for a conservative estimation of the RoA. Notice that the condition for inclusion in the RoA provided by Lemma \ref{th:clf}, is only valid when there is no model-plant mismatch. Thus, we have to incorporate the learned uncertainty terms from Section \ref{sec:04compoundkernel} as we do when we formulate the GP-CLF-SOCP in Section \ref{sec:05optimization}. 

\begin{theorem}
\label{theorem:learnedLyapCert}
Taking the GP posterior distribution from the training data $\{(x_j, u_j, z_j)\}_{j=1}^{N}$, and $\beta$ from \eqref{eq:sigmaUCB}, if there exists a $c > 0$ such that for all $x\in \Omega_{c} $ it holds that

\begin{equation}
\label{eq:learnedLyapCert}
    \inf_{u\in U}\ \tilde{\dot{V}}(x,u) + \mu_{*}(x,u) + \beta \sigma_{*}(x,u) < 0,
\end{equation}
then $\Omega_{c} $ is in the RoA with probability at least $(1-\delta)$.
\end{theorem}

\begin{proof}
Proof follows from Lemma \ref{th:clf} and Theorem \ref{theorem:DeltaUCB}.
\end{proof}

Notice that this certificate is ``conservative" in the sense that it takes the worst-case bound of the effect of the uncertainty term, based on the collected data. Therefore, if we collect more data and improve our GP model to have less uncertainty, then the conservatism will reduce and we will be able to obtain a bigger subset of the RoA. This is the central principle of the algorithm.

\subsection{Algorithm Overview}

Finally, we give an overview of the proposed algorithm. 
\subsubsection{Initial GP Model}
We start by considering a level set $\Omega_{c_0} $ which is small-enough to be a subset of the RoA (Fig \ref{fig:algorithm}.a). Such $c_0>0$ always exists due to our assumption that $V$ is a locally valid CLF.
We collect an initial batch of training data ($D_0$) from a set of trajectories whose initial states are randomly sampled from $\Omega_{c_0} $, and train an initial GP regression model. Here, we use the nominal model-based CLF-QP from \eqref{eq:clf-qp-all} as our stabilizing controller.

\subsubsection{Episodic Learning}
The main loop of our algorithm consists of a series of episodes, and each $i$-th episode is mainly composed of three steps. \textbf{1)} In the first step (Fig. \ref{fig:algorithm}.b), we obtain a set of $N_e$ points from $(\Omega_{ (c_{i-1} + \small{\Delta} c_i)} \setminus \Omega_{c_{i-1}} ) \times \mathcal{U}$ at which the variance of the posterior of the current GP model is maximal.
$\small{\Delta} c_i$ is the parameter that determines the size of the new exploration region. \textbf{2)} Next (Fig \ref{fig:algorithm}.c), we run short rollouts by taking each point from Step 1 as our initial state and initial control input. During the rollouts, we also evaluate the stabilizability condition \eqref{eq:learnedLyapCert} at each timestep. Note that such evaluation is a feasibility problem which is also an SOCP since \eqref{eq:learnedLyapCert} is a second-order cone constraint. After the rollouts, we expand the level of $V$ (we determine $c_{i}$) up to a point for which \eqref{eq:learnedLyapCert} becomes infeasible.
\textbf{3)} Finally (Fig \ref{fig:algorithm}.d), we add the data obtained from the trajectories within $\Omega_{c_i}$ to our data set, and train the next GP regression model.

\begin{remark}
In Step 2 of an episode, we check condition \eqref{eq:learnedLyapCert} only for finite sampled states in $\Omega_{c_i}  \setminus \Omega_{c_{i-1}}$, whereas Theorem \ref{theorem:learnedLyapCert} requires \eqref{eq:learnedLyapCert} to be satisfied at every state in $\Omega_{c_i}$. Notice that brute-force verification for the whole region of $\Omega_{c_i} $ will scale poorly with state dimension. Even though we do not have the rigorous guarantee of Theorem \ref{theorem:learnedLyapCert} with this algorithm, the error in the estimated $c_{max}$ does not affect the probabilistic guarantee of the resulting GP-CLF-SOCP controller.
In practice, we observe that we can well approximate $c_{max}$ such that $\Omega_{c_{max}}$ is contained in the true RoA (See Fig. \ref{fig:algorithm}(e)).
\end{remark}

%% file: 07results.tex
\section{Examples}
\label{sec:07results}

\subsection{Two-dimensional System: Inverted Pendulum}
\label{subsec:invertedpendulum}
Consider a control-affine two-dimensional inverted pendulum as the one in \cite{berkenkamp2017saferl}, with parameters of the plant $m_{\textrm{plant}}\!=\!2$kg, $l\!=\!1$m and for the model, $m_{\textrm{model}}\!=\!1$kg, $l\!=\!1$m, which results in model uncertainty in both $f$ and $g$ in \eqref{eq:system}.

A CLF-QP controller \eqref{eq:clf-qp-all} based on the nominal model is designed to stabilize the pendulum to the upright position.
In order to illustrate the effects of model uncertainty, we compare it with the CLF-QP controller based on the true plant dynamics. The difference between the two controllers (Fig.~\ref{fig:results-ip-case1}) is due to the effects of model uncertainty. Specifically, in this case the model uncertainty makes the system converge more slowly.

Fig.~\ref{fig:algorithm} depicts the data collection algorithm and the resulting training data for the GP model. The results of deploying the GP-CLF-SOCP controller, with a confidence level of $1\!-\!\delta\!=\!0.95$, are presented in Fig.~\ref{fig:results-ip-case1} in blue lines.
Note that the results are very similar to those from the CLF-QP based on the true plant dynamics, which means that the GP-CLF-SOCP successfully captures the correct effects of model uncertainty. Also, the computation time of the GP-CLF-SOCP, including the GP inference time, is $9.1\pm2.2$ms (max: $25.7$ms) on a laptop with a 10th-gen Intel Core i7 and 32GB RAM.

In order to benchmark the GP-CLF-SOCP, we compare its performance with the one obtained if we only learn the uncertainty in $f$, as done in previous works \cite{fan2019balsa, zheng2020learning}. For this, we design a GP-based Control Lyapunov Function Quadratic Program (GP-CLF-QP) that only learns the uncertainty $\Delta_1$ in \eqref{eq:mismatchaffine}. The results of this controller are also shown in Fig.~\ref{fig:results-ip-case1}.

\subsection{System with Multiple Control Inputs: Kinematic Bicycle}
\label{subsec:kinematicbicycle}

Next, in order to show that our method can be successfully applied to systems with higher state dimension and multiple control inputs, we apply it to track a reference trajectory using a kinematic bicycle model. The state is defined as $x = [p_x , p_y , v , \theta , \gamma]^T$ ($p_x, p_y$: position coordinates, $v$: speed, $\theta$: heading angle, $\gamma$: tangent of the steering angle). The dynamics of the system are given as
\vspace{-1em}

\small

\begin{equation}
    \dot{x} = f(x) + g(x)u, \quad f(x)\! =\! \begin{bmatrix} v \cos{\theta} \\ v \sin{\theta} \\ -f_{\mu} \\ v \gamma \\ 0 \end{bmatrix},\
    g(x)\! =\! \begin{bmatrix} 0 & 0 \\ 0 & 0 \\ b_v &  0 \\ 0 & 0 \\ 0 & b_{\gamma} \end{bmatrix},
\end{equation}

\normalsize

\noindent where $u\in\R^2$, and $f_\mu$, $b_v$, $b_\gamma$ are constants that emulate friction and skid effects. For the nominal model, we assume no such effects ($f_\mu\!=\!0$, $b_v\!=\!b_\gamma\!=\!1$) and for the plant, we use $f_\mu\!=\!1$, $b_v\!=\!1.5$, $b_\gamma\!=\!0.75$. The objective is to stabilize to a constant-velocity trajectory along the x-axis; $v(t)\!=\!5$, $p_y(t)\!=\!\theta(t)\!=\!\gamma(t)\!=\!0$. The initial state is set as $x_0 = [0,0.25,2,0.25,0.25]^T$.

\begin{figure}
\centering
\includegraphics[width=0.45\textwidth]{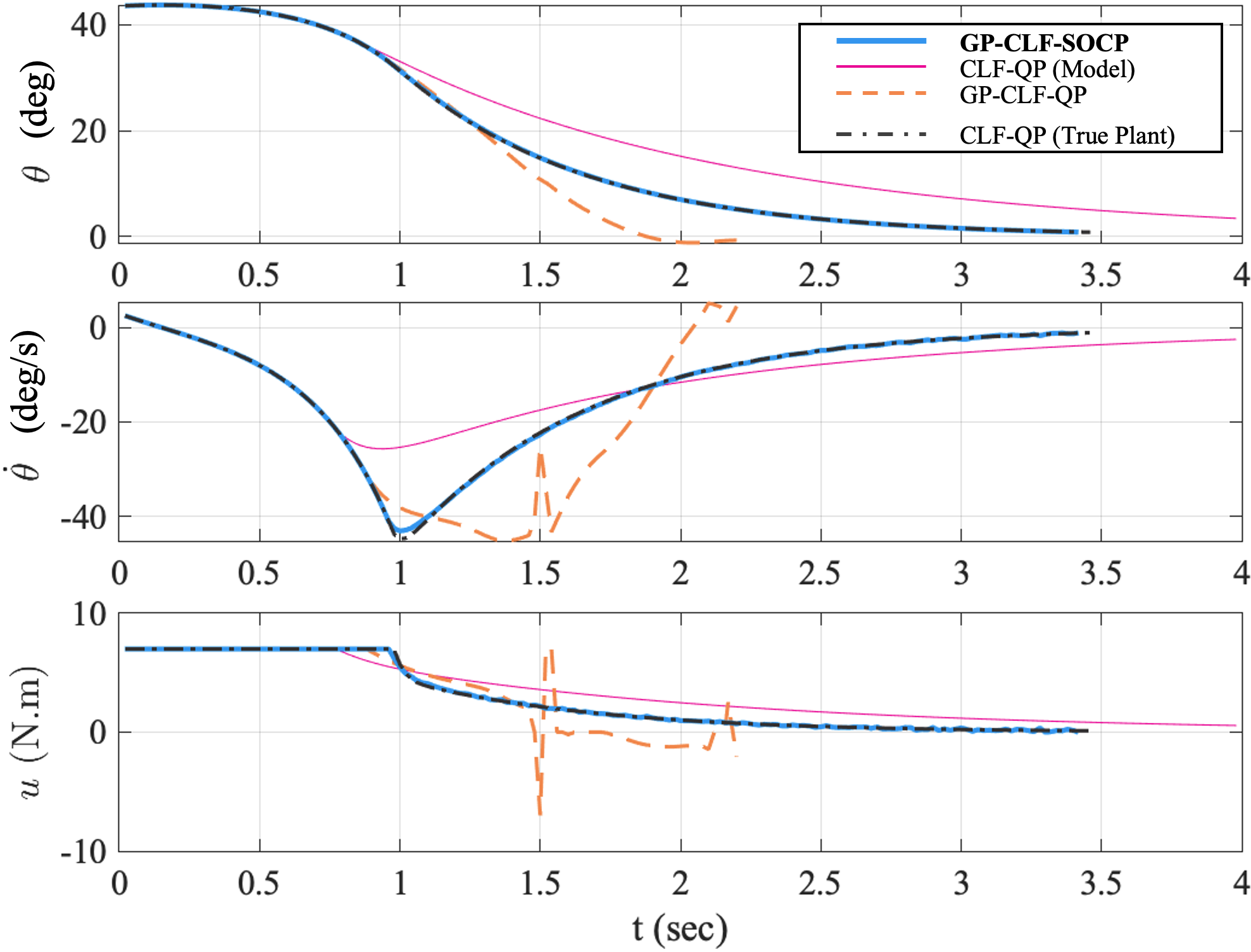}
\caption{Simulation results of applying the GP-CLF-SOCP to the inverted pendulum example, with a model-plant mismatch of $m_{\textrm{plant}}\!=$2kg, $m_{\textrm{model}}\!=$1kg, compared to the nominal-model-based CLF-QP, and to the GP-CLF-QP that does not consider the uncertainty affected by $u$. Results of the CLF-QP based on the true plant are also provided to show that the GP-CLF-SOCP learns the correct exponential CLF constraint.
}\label{fig:results-ip-case1}
\vspace{4pt}
\end{figure}

Fig. \ref{fig:results-bicycle} shows the simulation results of the GP-CLF-SOCP and those of the CLF-QP based on the nominal model and the true plant. Here, we use a polynomial CLF \cite{clfdemonstration}, which is verified to be locally stabilizing for the nominal model. While the nominal model-based CLF-QP oscillates around the reference trajectory, the GP-CLF-SOCP successfully converges to the reference trajectory. 

%% file: 08conclusion.tex
\section{Conclusion}
We have presented a method to design a stabilizing controller for control-affine systems with both state and input-dependent model uncertainty using GP regression. For this purpose, we have proposed the novel ADP compound kernel, which captures the control-affine nature of the problem. This permits the formulation of the so-called GP-CLF-SOCP, which is solved online to obtain an exponentially stabilizing controller with probabilistic guarantees. After testing it on the numerical simulation of two different systems, we obtain a clear improvement with respect to the nominal model-based CLF-QP and we are able to closely match the performance of the true plant-based controller. 

\begin{figure}
\centering
\includegraphics[width=0.45\textwidth]{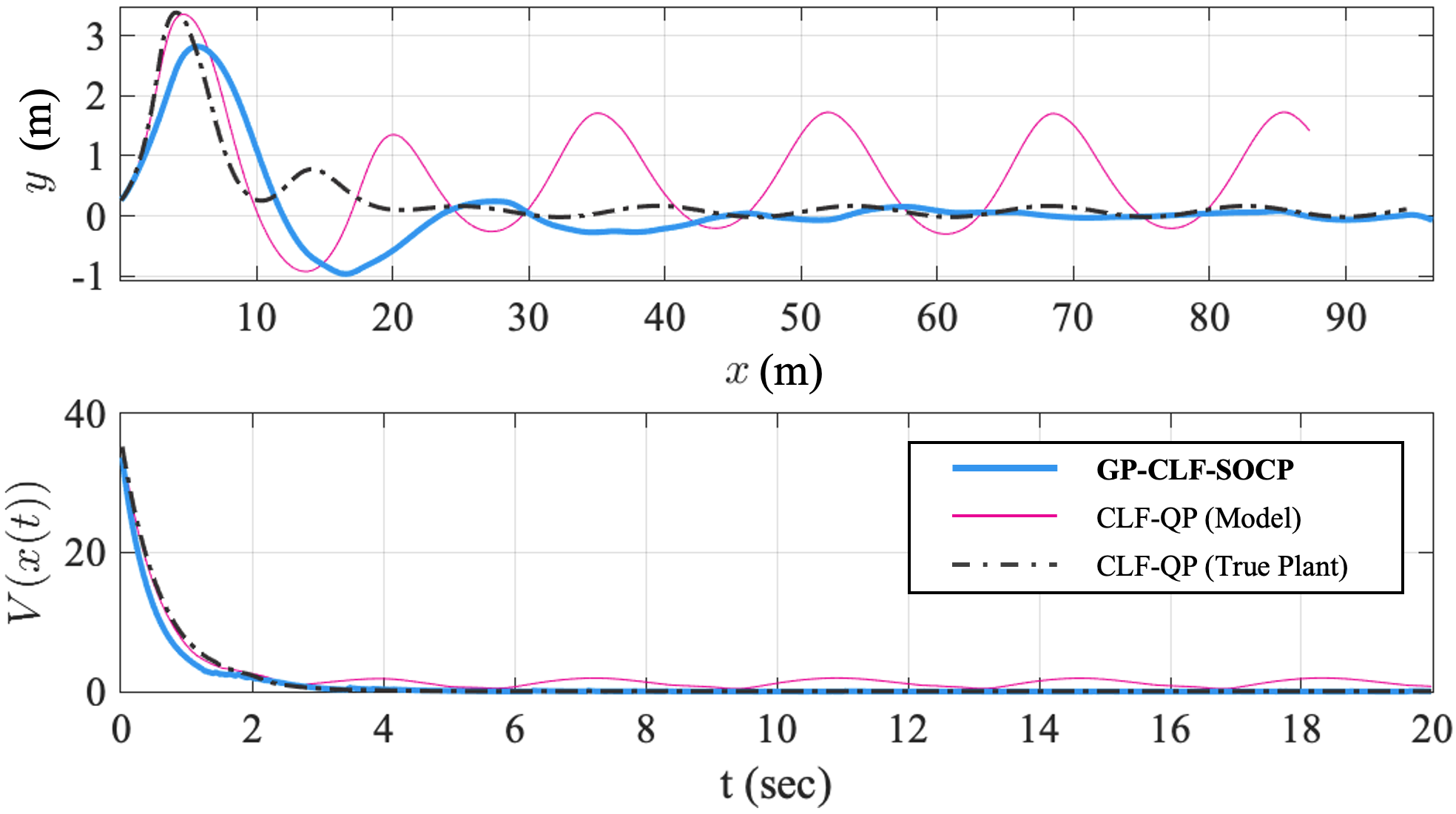}
\caption{Trajectories in $x-y$ plane (Top) and histories of $V(x(t))$ (Bottom) of the kinematic vehicle under artificial drift and friction to illustrate the applicability of the GP-CLF-SOCP to multi-input systems. Comparison between GP-CLF-SOCP, CLF-QP(Model), and CLF-QP(Plant). The sampling time is set as 20ms, and the comptutation time of the GP-CLF-SOCP per timestep is $10.3\pm1.9$ms (max: $20.4$ms). Number of training data points for GP-CLF-SOCP: 961.
}
\label{fig:results-bicycle}
\vspace{-10pt}
\end{figure}

\label{sec:08conclusion}